\documentclass[runningheads]{llncs}
\usepackage[T1]{fontenc}
\usepackage{graphicx}
\usepackage{dcolumn}
\usepackage{bm}
\usepackage{multirow}
\usepackage{physics}
\usepackage{hyperref}
\usepackage{xcolor}
\usepackage[normalem]{ulem}
\usepackage[linesnumbered,ruled,vlined]{algorithm2e}
\usepackage[caption=false]{subfig}
\usepackage{units}
\usepackage{amsfonts}
\usepackage{cite}
\usepackage{nopageno}

\begin{document}
\title{Tucker iterative quantum state preparation}
%
%
\author{Carsten Blank\inst{1}\orcidID{0000-0003-3450-0823} \and
Israel F. Araujo\inst{1}\orcidID{0000-0002-0308-8701}}
\authorrunning{Blank et al.}
%
\institute{data cybernetics ssc GmbH, Landsberg am Lech 86899, Germany
\email{blank@data-cybernetics.com}\\
\url{https://data-cybernetics.com}}
\maketitle              
\begin{abstract}
Quantum state preparation is a fundamental component of quantum algorithms, particularly in quantum machine learning and data processing, where classical data must be encoded efficiently into quantum states. Existing amplitude encoding techniques often rely on recursive bipartitions or tensor decompositions, which either lead to deep circuits or lack practical guidance for circuit construction. In this work, we introduce Tucker Iterative Quantum State Preparation (Q-Tucker), a novel method that adaptively constructs shallow, deterministic quantum circuits by exploiting the global entanglement structure of target states. Building upon the Tucker decomposition, our method factors the target quantum state into a core tensor and mode-specific operators, enabling direct decompositions across multiple subsystems.

\end{abstract}
\section{Introduction}

Quantum computing has emerged as a rapidly advancing field with significant recent progress in both theoretical and experimental domains. Current quantum hardware resides in the so-called Noisy Intermediate-Scale Quantum (NISQ) era, characterized by devices that operate with a limited number of qubits and are susceptible to noise and decoherence. These limitations pose substantial challenges in the engineering, control, and deployment of quantum systems. In particular, NISQ devices lack full fault tolerance and require sophisticated quantum error mitigation or correction strategies. Additional complications arise in the design and implementation of quantum circuits, especially concerning the reliable execution of quantum gates. A critical bottleneck in practical quantum computation is the task of encoding classical data into quantum states -- a process referred to as \textit{quantum state preparation}, or more generally, data encoding and loading.

Quantum state preparation is the process of encoding classical data into quantum states, which serve as inputs for quantum algorithms in applications such as quantum machine learning and data processing tasks~\cite{blank2020quantum,schuld2018supervised,park2019circuit}. Multiple approaches have been developed to address the quantum state preparation problem. Among the most widely used are \textit{Amplitude Encoding}~\cite{bergholm2005quantum,araujo2021divide,araujo2021configurable,araujo2024lowrank}, \textit{Qubit Encoding} (also referred to as \textit{Angle Encoding})~\cite{grant2018hierarchical}, and \textit{Hamiltonian Encoding}~\cite{cade2020strategies,schuld2018supervised}.

In this work, we introduce a novel amplitude encoding method named \textit{Tucker Iterative Quantum State Preparation} (Q-Tucker). This technique aims to provide a deterministic approach to state preparation by generating quantum circuits whose depth is adaptively determined by the entanglement structure of the target quantum state~\cite{araujo2024lowrank}.

Entanglement has long been recognized as a fundamental resource in quantum computing, enabling capabilities that are unattainable in classical computation. Nevertheless, quantum features such as entanglement and superposition can, to a certain extent, be simulated on classical hardware. In practice, the preparation of an $n$-qubit quantum state -- represented by a complex vector of $N = 2^n$ amplitudes -- for execution on current NISQ-era quantum processing units (QPUs) often depends on classical pre-processing.

Delegating these preparatory computations to classical resources becomes increasingly challenging as the dimensionality of the quantum state grows, a scenario commonly encountered in quantum machine learning and quantum data processing tasks. A well-established strategy for mitigating this complexity involves decomposing the target state into approximately low-entangled subsystems. Even approximate decompositions can offer substantial advantages, such as reduced quantum circuit depth and, in many cases, improved fidelity when compared to exact state initialization. As a result, efficiently characterizing and exploiting the entanglement structure of quantum states is a critical objective in the design of scalable quantum algorithms and state preparation protocols.

\subsection{Related work}

Several previous studies have explored low-rank quantum state preparation by leveraging hierarchical decompositions of multipartite quantum systems. Araujo et al. introduced an iterative scheme based on successive applications of the Schmidt decomposition to identify and disentangle bipartitions within a quantum system~\cite{araujo2024lowrank}. In their method, at each iteration, a suitable bipartition is selected, and the entanglement is reduced until no further disentanglement is possible, yielding a compact, low-rank representation of the original quantum state.

Regarding quantum state encoding, a variety of approaches have been proposed to map the $N$ complex-valued amplitudes of an arbitrary quantum state onto a physical quantum circuit. One of the earliest deterministic techniques was presented by Grover et al., who proposed a method to construct quantum superpositions corresponding to efficiently integrable probability distributions~\cite{grover2002superposition}. While conceptually straightforward, this method requires $\mathcal{O}(N)$ quantum operations, rendering it impractical for large-scale quantum systems.

An alternative family of methods exploits tensor network structures, particularly the matrix product state (MPS) representation, to facilitate efficient state encoding. Schön et al. proposed a sequential scheme to construct MPS representations via one- and two-qubit quantum gates, enabling scalable circuit generation for a broad class of quantum states~\cite{schon2006sequential}. Owing to their strong approximation capabilities, MPS methods have become widely adopted in quantum many-body physics. Building upon this foundation, subsequent works have integrated parameterized tensor-network models to reduce quantum circuit depth and to support hybrid quantum-classical training schemes~\cite{ran2020mps}. In some instances, iterative application of so-called \textit{disentangler} operators to MPS representations has been proposed to further improve state compression and circuit efficiency.

A conceptually distinct approach based on the Tucker decomposition has been proposed by Protasov et al.~\cite{protasov2021tucker}. In contrast to methods relying on the Schmidt decomposition or singular value decomposition (SVD), the Tucker decomposition generalizes these techniques by enabling one-versus-all decompositions across multiple modes.
While the authors advocate for the use of the Tucker decomposition in quantum state preparation~\cite{protasov2021tucker}, their approach does not specify a selection criterion for determining the appropriate decomposition among the many non-unique possibilities.
Finally, the work provides no concrete procedure for mapping the decomposed representation into a quantum circuit, leaving a critical gap between theoretical formulation and practical implementation.

\subsection{Objectives and organization}

The primary objective of this work is to address the limitations of existing state preparation techniques by introducing a method based on tensor decomposition that is naturally suited for multipartite quantum systems. Unlike prior approaches that rely on recursive bipartitions -- often resulting in complex hierarchical structures -- our method circumvents these constraints by directly exploiting the global structure of the quantum state.

The remainder of the paper is organized as follows.
Section~\ref{sec:qtucker} introduces the Tucker Iterative Process, detailing its decomposition strategy, fidelity control mechanisms, and the circuit synthesis procedure. Section~\ref{sec:convergence} presents an analysis of the method’s convergence. Section~\ref{sec:results} reports the experimental results. Finally, Section~\ref{sec:conclusion} concludes the paper with a summary of the findings and future research directions.

\section{Tucker iterative quantum state preparation}
\label{sec:qtucker}

This section introduces a hardware-efficient method for encoding classical data into a quantum device using an iterative tensor decomposition strategy. Using a multipartite-specific decomposition based on the Tucker tensor representation, the method addresses the limitations of previous methods~\cite{araujo2024lowrank} and eliminates the need for complex hierarchical bipartitions. The quantum circuit is constructed from the Tucker representation, with a performance metric guiding the iterative refinement process (see Fig.~\ref{fig:qtucker})

Any vector \( v \in \mathbb{C}^N \) can be represented as a \(m\)-tensor \( T \in \mathbb{C}^{d_1} \otimes \ldots \otimes \mathbb{C}^{d_m} \), for an arbitrary \( m > 1 \), with the constraint \( \prod_i d_i = N \). Given such a tensor \( T \), a representation or decomposition can be determined. One such representation is the Tucker decomposition:  
\[ T = G \times_1 W^{(1)} \times_2 W^{(2)} \times_3 \ldots \times_m W^{(m)}, \]  
where \( \times_i \) denotes the mode-\(i\) product, \( G \in \mathbb{C}^{p_1} \otimes \ldots \otimes \mathbb{C}^{p_m} \) is the core tensor (with \( p_i \le d_i \)), and the factor matrices \( W^{(i)} \in \mathbb{C}^{d_i \times p_i} \).
This decomposition is naturally suited for quantum computing since the tensor decomposition of \( T \) directly corresponds to a quantum state decomposition:  
\[ |\psi\rangle = (W^{(1)} \otimes \ldots \otimes W^{(m)})|G\rangle, \]  
where we denote the operators \( W^{(i)} \) using the same symbol as the factor matrices, even if their dimensions do not match (\( p_i \neq d_i \)), in which case columns are padded with basis vectors. The core tensor \( G \) is interpreted as a quantum state \( |G\rangle \) via vectorization.

A key feature is that the factor matrices can be interpreted as unitary operators acting on separate subspaces (hence, the Tucker format is also known as the \textit{subspace format}). From a quantum computing perspective, these unitary operations can be executed in parallel -- if hardware permits -- making the depth of any resulting quantum circuit determined by the longest individual factor circuit \( W^{(i)} \). The geometric entanglement of the core tensor is never greater than that of the original state and can, in many cases, be significantly lower. This property enables a recursive or iterative approach that terminates. If recursion halts after \( r \) steps, the resulting decomposition is:  
\[ |\psi\rangle = W_1 \ldots W_r |0\rangle. \]

\begin{figure*}[t]
    \centering
    \includegraphics[width=1.0\linewidth]{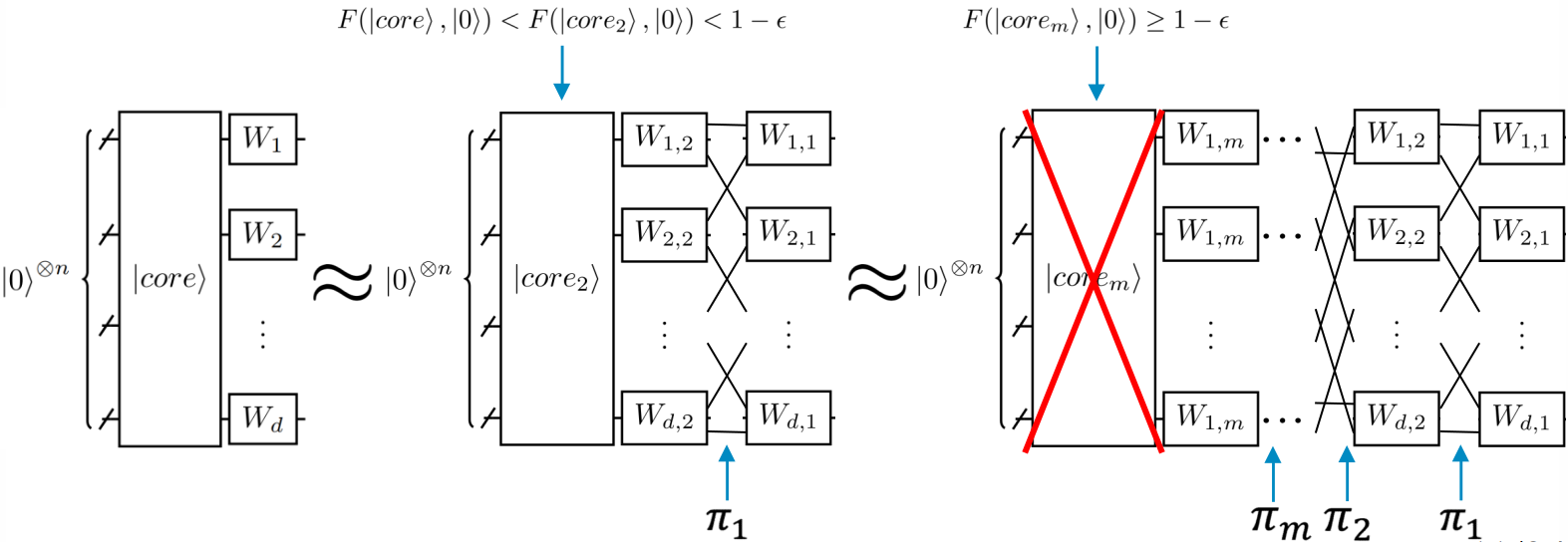}
    \caption{Quantum circuit representation of the Tucker Iterative Process for circuit compression. The initial quantum state is decomposed into a core vector \(|\text{core}\rangle\) and local hardware-efficient unitary (or isometry) factors \(W_i\). At each iteration, the fidelity \(F(|\text{core}_j\rangle, |0\rangle)\) is evaluated to determine whether the core can be approximated by the ground state. If the fidelity exceeds the threshold (\(F \geq 1 - \epsilon\)), the core is removed (indicated by the red cross). 
    Between iterations, logical qubit swaps \(\pi_j\) are applied to optimize the placement of the factors. These swaps are guided by minimizing the bond dimension between qubit partition blocks -- each corresponding to a factor \(W_i\) -- which reflects the entanglement between subsystems.}
    \label{fig:qtucker}
\end{figure*}

Each iteration \( j \) of the process yields a factorization \( W_j = W_{(j,1)} \otimes \ldots \otimes W_{(j,m)} \). If the resulting core tensor -- or quantum state \( |G_r\rangle \) -- can be expressed as a product state (i.e., composed solely of operations that can be absorbed into the adjacent factors), or if the fidelity between the core state and the true ground state exceeds a predefined threshold, then the core state is adopted as the ground state. In this case, it requires no additional quantum operations. Each factor \( W_{(j,i)} \) is implemented as a unitary or isometry operator \( W^{(i)} \in \mathbb{C}^{d_i \times p_i} \). If a full unitary operator is not required (\( p_i < d_i \)), then isometries can be used -- an important detail that enables circuit simplification and motivates the method described below.

It is well known that the tensorization of a vector \( v \in \mathbb{C} \) into a tensor \( T \in \mathbb{C}^{d_1} \otimes \ldots \otimes \mathbb{C}^{d_m} \) is not unique. Prior work has shown that different sequences of logical swap operations can lead to different entanglement behaviors~\cite{araujo2024lowrank}. This result extends to tensorization itself, resulting in a search space defined by the number of modes \(m\), the dimension tuple \(\mathbf{d} = (d_{i_1}, \ldots, d_{i_m})\), and a permutation \(\pi \in \text{Sym}([m])\) that encodes the swap pattern. Each \textit{tensorization configuration} is thus characterized by a tuple \((m, \mathbf{d}, \pi)\). Since the optimal configuration \((m, \mathbf{d}, \pi)\) is not unique, it is possible to seek those that balance approximation quality with implementation complexity.

In general, there exists a configuration \((m, \mathbf{d}, \pi)\) such that the entanglement in the core tensor is smaller than in the original tensor \( T \). This allows the method to be applied recursively until a product state is reached, which can be trivially initialized. Alternatively, methods such as that of Rudolph \& Grover~\cite{grover2002superposition}, or their concrete implementations~\cite{bergholm2005quantum,mottonen2004quantum}, can be employed.

The search can be simplified by taking into account the coupling map limitations of the target device. On currently available quantum hardware, efficient two-qubit operations are typically restricted to neighboring qubits. Under this constraint, we can fix \( m = n/2 \), leading to \( d_i = 4 \). Since the dimension of these subsystems is fixed, the cost of each SVD is constant.
This eliminates the scalability bottleneck often associated with tensor decompositions such as Tucker, where SVD costs scale cubically with the input dimension.


Even with such simplifications, the search process still depends on logical qubit swaps, which are costly. Identifying the partition with the minimal bond dimension -- that is, the one producing the core vector with the least entanglement -- requires a combinatorial search. The complexity of this search is determined by the size of the partition blocks, \( n_i=\log_2 d_i \), and grows according to the binomial coefficient \( \binom{n}{n_i} \). In the worst-case scenario, where all blocks have equal size, the complexity becomes \( \mathcal{O}(2^n) \). For large systems, the computational overhead of evaluating the bond dimension at each step becomes prohibitive, making entanglement quantification infeasible.

To address this, heuristic strategies and physically motivated constraints are essential to guide the search process. These include limiting the partition search space to locality-preserving cuts, prioritizing low-overlap subsystems, or leveraging prior knowledge about the structure of the target state (e.g., symmetries or sparsity). Such strategies drastically reduce the search space and enable approximate yet effective decomposition, balancing computational tractability with circuit efficiency.

\section{Convergence guarantee of Q-Tucker}
\label{sec:convergence}

Here we formalize the correctness of the iterative procedure described in Section~\ref{sec:qtucker}. Throughout, an \emph{oracle} provides a partition $\mathcal P=\{B_1,\dots,B_m\}$ (a qubit set per block) that maximizes ``inner'' entanglement within blocks (equivalently, minimizes inter-block entanglement or bond
dimension). Given this $\mathcal P$, we run a Tucker step and then \emph{fix the gauge} (``monotone gauge'') by choosing the first column of each block operator to maximize the block-product overlap with the current state. Tucker/HOSVD and the non-uniqueness of factors up to rotations inside singular subspaces are classical facts~\cite{deLathauwer2000HOSVD,kolda2009tensor}.

\paragraph{Notation and per-step objective.}
Let $\lvert G_{j-1}\rangle$ be the core state before iteration $j$ and $\lvert 0\rangle=\bigotimes_i\lvert 0\rangle_{B_i}$ the block-wise ground state. For a partition $\mathcal P$, define the multipartite product-overlap
\[
\alpha_{\mathcal P}(\lvert\phi\rangle)\;:=\;\max_{\|u_i\|=1}\bigl|\langle u_1\!\otimes\!\cdots\!\otimes\!u_m \mid \phi\rangle\bigr|,
\]
which is called the \emph{entanglement eigenvalue}, and corresponds to the closest separable state~\cite{wei2003geometric}.
The \emph{monotone gauge} sets the first column of each block unitary so that $W^{(i)}\lvert 0\rangle_{B_i}=u_i^\star$ where $(u_i^\star)$ attains the maximum above.

\begin{lemma}[monotone-gauge identity]
For the (oracle-provided) partition $\mathcal P_j$ used at iteration $j$,
\[
\langle 0\mid G_j\rangle \;=\; \alpha_{\mathcal P_j}(\lvert G_{j-1}\rangle),
\qquad
F_j \;=\; \alpha_{\mathcal P_j}(\lvert G_{j-1}\rangle)^2.
\]
\end{lemma}

\begin{proof}
By construction, $(\otimes_i W_j^{(i)\dagger})\lvert G_{j-1}\rangle=\lvert G_j\rangle$ and $W_j^{(i)}\lvert 0\rangle_{B_i}=u_i^\star$; hence
$\langle 0\mid G_j\rangle=\langle \otimes_i u_i^\star \mid G_{j-1}\rangle$, which is the maximizing value.
\end{proof}

\begin{corollary}[per-step monotonicity for any partition]
For any partition $\mathcal P$,
\[
\alpha_{\mathcal P}(\lvert G_{j-1}\rangle)\;\ge\; \bigl|\langle 0\mid G_{j-1}\rangle\bigr|,
\]
since $\lvert 0\rangle$ is a feasible block-product vector. Therefore, with the monotone gauge on $\mathcal P_j$,
\[
F_j \;=\; \alpha_{\mathcal P_j}(\lvert G_{j-1}\rangle)^2 \;\ge\; \bigl|\langle 0\mid G_{j-1}\rangle\bigr|^2 \;=\; F_{j-1}.
\]
This mirrors the monotonicity of ALS/HOOI-style block
updates for Tucker models~\cite{deLathauwer2000BestRank1,kolda2009tensor,tseng2001convergence}.
\end{corollary}

\paragraph{Role of the oracle.}
The oracle’s partition, which maximizes inner entanglement, is used to \emph{accelerate} progress: it typically enlarges the cut-based ceiling for block-product overlaps (see below). However, monotonicity $F_j\ge F_{j-1}$ holds \emph{for any partition} when the monotone gauge is applied.

\paragraph{Cut ceiling and strict improvement.}
For a partition $\mathcal P$, let $\beta(\mathcal P):=\min_{C}\lambda_{\max}(C)$ be the minimum, over block cuts $C$, of the largest Schmidt coefficient of $\lvert G_{j-1}\rangle$ across $C$. For any bipartition, the largest Schmidt coefficient equals the maximal overlap with a product state across that cut~\cite{Nielsen_Chuang_2010,plenio2007entanglement}; therefore any block-product is bounded by
\[
F_j \;=\; \alpha_{\mathcal P_j}(\lvert G_{j-1}\rangle)^2 \;\le\; \beta(\mathcal P_j).
\]
Partitions maximizing inner entanglement tend to increase $\beta(\mathcal P_j)$, raising a tight upper bound on $F_j$. A \emph{strict} increase $F_j>F_{j-1}$ is guaranteed whenever the ``worst'' cut for $\mathcal P_j$ is unique and its top Schmidt vectors factor across blocks, in which case $\alpha_{\mathcal P_j}^2=\beta(\mathcal P_j)$ and the larger ceiling is attained.

\paragraph{Stalls and block-size expansion.}
Fix a maximum block size $k$ (so each block unitary acts on $\le k$ qubits). If $F_j=F_{j-1}$ persists, we say the process \emph{stalls at level $k$}. This occurs because, with fixed $k$, each step (via the monotone gauge) already attains the \emph{best possible} block-product overlap within the current model class, i.e., $F_j=\alpha_{\mathcal P_j}(\lvert G_{j-1}\rangle)^2$. When the state’s overlap is capped by the \emph{cut ceiling} for the chosen partition ($\alpha_{\mathcal P_j}=\beta(\mathcal P_j)$), or when all admissible partitions yield the same ceiling (or the worst-cut top Schmidt vectors fail to factor across blocks), no strict improvement is possible. Increasing $k$ strictly enlarges the feasible set and (weakly) raises the ceilings $\beta$, so the limiting fidelity cannot decrease; allowing $k$ to reach $n$ (a single block) makes $F=1$ attainable.

\begin{theorem}[global convergence, stall-and-grow]
\label{thm:global-convergence}
At each iteration $j$, (i) receive only the partition $\mathcal P_j$ from the oracle; (ii) run a Tucker step on $\mathcal P_j$ and apply the monotone gauge to form $\lvert G_j\rangle$; (iii) if stalled at block size $k$, increase $k$. Then $(F_j)$ is non-decreasing and converges for fixed $k$. Moreover, under stall-and-grow with $k$ allowed to reach $n$ (one block), the process attains $F=1$ (exact preparation) in finitely many expansions.
\end{theorem}

\begin{proof}
Monotonicity and boundedness ($0\le F_j\le 1$) yield convergence for fixed $k$. Enlarging $k$ cannot reduce the achievable limit; with $k=n$ any pure state can be mapped to $\lvert 0\rangle$ by a single block unitary~\cite{Nielsen_Chuang_2010}, whence $F=1$.
\end{proof}

\paragraph{Summary.}
Q-Tucker’s per-iteration fidelity \emph{never decreases} when the monotone gauge is enforced after the oracle’s partition. Oracle partitions that maximize inner entanglement typically improve the \emph{rate} of progress by enlarging the cut ceiling. If the process stalls at a given block size, raising the block size restores progress; allowing it to reach the full register guarantees exact convergence.

\section{Correlation Graph guided partition search}
\label{sec:corrgraph}

The search over tensorization configurations \((m,\mathbf d,\pi)\) (cf.\ Sec.~\ref{sec:qtucker}) can be formulated as a graph-partitioning problem informed by the entanglement structure of the target state. We introduce the \emph{Correlation Graph}, a heuristic construct that approximates the state’s entanglement configuration by mapping qubit correlations onto weighted edges. The resulting cut objective encodes the desired trade-off: maximize intra-block correlation (high ``inner'' entanglement) while minimizing inter-block correlation (low bond dimension between blocks).

\subsection{Graph definition and weights}

Let $\ket{\psi}\in(\mathbb C^2)^{\otimes n}$ be a normalized pure state. Define a weighted, undirected graph
\(
G=(V,E,W),
\)
where $V=\{1,\dots,n\}$ (one node per qubit) and $W=[w_{ij}]$ encodes pairwise correlation strengths. Several definitions of $w_{ij}$ are possible; in practice we employ a two-qubit correlation proxy derived from reduced density matrices,
\begin{equation}
\label{eq:wij}
w_{ij}\;:=\;\underbrace{S(\rho_i)+S(\rho_j)-S(\rho_{ij})}_{\text{mutual information } I(i\!:\!j)}
\quad\text{or}\quad
w_{ij}\;:=\;\|\rho_{ij}-\rho_i\!\otimes\!\rho_j\|_F,
\end{equation}
where $\rho_{ij}=\Tr_{\overline{ij}}(\ket{\psi}\bra{\psi})$ and $\rho_i=\Tr_{\overline{i}}(\ket{\psi}\bra{\psi})$. Both measures vanish for product pairs and increase with correlation strength.\footnote{Any monotone, nonnegative symmetric proxy may be used; we typically select the Frobenius option for numerical robustness and ease of accumulation.}

\paragraph{Fast accumulation.}
Let $N=2^n$. Building all one-qubit marginals $\{\rho_i\}_{i=1}^n$ by reshaping the state tensor once and forming $n$ Gram products costs $O(N\,n)$. Constructing the \emph{fully connected} correlation graph then requires all two-qubit marginals $\rho_{ij}$ for $i<j$, each obtainable as a $(4\times 2^{n-2})$ Gram product after an axis move; this yields $O(N)$ per pair and a total of $O(N\,\binom{n}{2})=O(N\,n^2)$ time, with $O(N)$ extra memory (no $N\times N$ density is ever materialized). In practice, we often evaluate a \emph{sparsified surrogate} consistent with hardware coupling or a small $k$-NN degree cap $d_{\max}$, which reduces the two-qubit work to $O(N\,|E|)=O(N\,n)$ when degrees are bounded. This preserves the paper’s linear-in-\#amplitudes scaling target while remaining faithful to the fully connected definition (the surrogate is a computational shortcut, not a change in the model). We reuse a single tensor view of $\psi$ across all marginals and cache $S(\rho_i)$ for mutual-information weights. For larger blocks ($n_i>2$), exact higher-order marginals (e.g., $\rho_{ijk}$) scale as $O(N\,n^3)$ and beyond, which is why we keep \emph{pairwise} weights for steering and increase $n_i$ only upon stalls (Sec.~\ref{sec:convergence}).

\subsection{Heuristic approximation and motivation}

The Correlation Graph constitutes a heuristic surrogate for the true entanglement layout of $\ket{\psi}$. It approximates a fundamentally multipartite phenomenon through pairwise summaries, which leads to an inevitable approximation gap:

\begin{itemize}
  \item \textbf{Higher-order correlations.} Multipartite entanglement (e.g., GHZ- or W-type) cannot, in general, be reconstructed from pairwise marginals. Pairwise measures may underestimate nonlocal correlations that appear only at the level of three or more qubits.
  \item \textbf{Cut vs.\ edge sum.} The bond dimension relevant to Tucker/HOSVD updates depends on the full Schmidt spectra across block cuts, whereas the graph cut in~\eqref{eq:cut} aggregates pairwise proxies. These quantities correlate well empirically but are not identical.
  \item \textbf{Basis and model dependence.} Working in the computational basis and with fixed block size $n_i$ introduces a coarse-graining that can under- or over-emphasize specific correlations.
  \item \textbf{Sparsification.} For tractability we cap node degree and restrict edges to physically adjacent (or near-adjacent) pairs, omitting weak long-range links that may carry small but nonzero entanglement.
\end{itemize}

Nevertheless, the heuristic is \emph{well aligned} with Q-Tucker’s update mechanism: since each iteration manipulates small local blocks via constant-size SVDs, a pairwise map of where correlations concentrate provides a practical and effective steering signal for forming blocks that reduce inter-block coupling.

An exact optimization of the partition would require evaluating entanglement across all possible cuts and permutations $(m,\mathbf d,\pi)$, an intractable task even for small systems:
\begin{enumerate}
  \item The number of partitions grows super-exponentially (Bell numbers); even restricting to pairwise blocks, the count is $n!/(2^{n/2}(n/2)!)$, prohibitive for $n>10$.
  \item Each evaluation involves computing Schmidt spectra or running full Tucker steps to determine bond dimensions.
  \item Globally minimizing entanglement between all blocks is equivalent to known NP-hard tensor-rank problems.
\end{enumerate}
Hence, the Correlation Graph serves as an efficient, physically motivated heuristic to approximate this otherwise unmanageable optimization.

\subsection{Partition objective and relation to bond dimension}

Given a partition $\mathcal P=\{B_1,\dots,B_m\}$ of $\{1,\dots,n\}$ into blocks of size $n_i=\log_2 d_i$, define the inter-block cut
\begin{equation}
\label{eq:cut}
\Phi(\mathcal P)\;:=\;\sum_{\substack{i\in B_a,\; j\in B_b\\ a\neq b}}\! w_{ij}.
\end{equation}
Minimizing $\Phi(\mathcal P)$ groups strongly correlated qubits within the same block. Empirically, smaller inter-block correlation correlates with lower Schmidt ranks across the corresponding cuts and, consequently, smaller bond dimensions in the Tucker core. In the Q-Tucker iteration, partitions with small $\Phi$ enlarge the per-step ceiling for block-product overlaps (Sec.~\ref{sec:convergence}), leading to faster convergence and shallower circuits.

\subsection{Search routine and variable block sizes}

Although the default implementation fixes $n_i=2$ (two-qubit blocks), the same procedure generalizes naturally to $n_i>2$. Increasing block size allows each SVD to capture more local correlations at the expense of higher classical cost and more expensive unitary synthesis. In practice, $n_i$ is incrementally raised whenever the algorithm \emph{stalls} (cf.\ Theorem~\ref{thm:global-convergence}), providing a natural stall-and-grow mechanism consistent with the convergence proof.

{\small
\begin{algorithm}[H]
\DontPrintSemicolon
\caption{Correlation Graph guided partition ($n_i{=}2$)}
\label{alg:corr-graph}
\KwIn{Statevector amplitudes $\psi$; block size $n_i{=}2$; metric for $w_{ij}$ (Sec.~\ref{sec:corrgraph}); optional constraint graph $\mathcal C$}
\KwOut{Partition $\mathcal P=\{\{i_1,j_1\},\dots,\{i_{n/2},j_{n/2}\}\}$}
\Begin{
  \tcp{(1) Build fully connected weights (Sec.~\ref{sec:corrgraph})}
  Compute all one-qubit marginals $\{\rho_i\}$ and all two-qubit marginals $\{\rho_{ij}\}_{i<j}$; assemble $W=[w_{ij}]$.\;

  \tcp{(2) Choose edge set}
  \eIf{$\mathcal C$ is not provided}{
    $E \leftarrow \{(i,j)\mid 1\le i<j\le n\}$ \tcp*{fully connected}
  }{
    $E \leftarrow \{(i,j)\in \mathcal C\}$ \tcp*{constrained/sparse edges}
  }

  \tcp{(3) Exact maximum-weight perfect matching when feasible}
  Attempt to compute a maximum-weight perfect matching $M^\star$ on graph $(V,E)$ with weights $w_{ij}$ (e.g., Edmonds/Blossom).\;
  \If{$M^\star$ is a perfect matching}{
    \Return{$\mathcal P \leftarrow \{\{i,j\}:(i,j)\in M^\star\}$} \tcp*{optimal for given $E$}
  }

  \tcp{(4) Heuristic fallback if no perfect matching exists on $E$}
  Initialize $M \leftarrow \emptyset$, $U \leftarrow V$.\;
  \While{$U\neq\emptyset$}{
     Pick $(i,j)\in E$ with $i,j\in U$ of maximum $w_{ij}$; set $M\leftarrow M\cup\{(i,j)\}$; $U\leftarrow U\setminus\{i,j\}$.\;
     \If{no such $(i,j)$ exists}{
        \tcp{augment $E$ minimally: allow top-$k$ missing edges incident to $U$}
        Add to $E$ the highest-weight missing edges touching $U$ (small $k$).\;
     }
  }
  \tcp{Local 2-opt refinement (pairwise exchanges)}
  \Repeat{no improvement}{
    For each two disjoint pairs $\{i,k\},\{j,\ell\}\in M$, test the four swaps
    $i\leftrightarrow j$, $i\leftrightarrow \ell$, $k\leftrightarrow j$, $k\leftrightarrow \ell$ and apply the one that maximizes $\sum_{(a,b)\in M} w_{ab}$.\;
  }
  \Return{$\mathcal P \leftarrow \{\{i,j\}:(i,j)\in M\}$}\;
}
\end{algorithm}
}

Greedy matching exploits that, for $n_i{=}2$, the partition reduces to a maximum-weight matching problem whose objective is equivalent to minimizing~\eqref{eq:cut}. The local refinement stage corrects greedy suboptimalities at negligible cost.

\subsection{Integration and complexity}

At iteration $j$:
\begin{enumerate}
  \item Estimate $W$ from $\ket{G_{j-1}}$ or from the current approximation (we use $\ket{G_{j-1}}$ for stability).
  \item Run Alg.~\ref{alg:corr-graph} (or its $n_i>2$ variant) to obtain $\mathcal P_j$ and, optionally, the logical permutation $\pi_j$ compatible with the hardware coupling map $\mathcal C$.
  \item Execute a Tucker step on $\mathcal P_j$ and apply the monotone gauge (Sec.~\ref{sec:convergence}) to form $\ket{G_j}$.
  \item If $F_j$ stalls, either (a) re-run Alg.~\ref{alg:corr-graph} with a larger $d_{\max}$ or perturbed weights, or (b) increase block size $n_i$ (stall-and-grow).
\end{enumerate}

With bounded-degree accumulation and fixed block size $n_i=2$, weight construction is $O(N\,|E|)$ 
(linear in the number of amplitudes for bounded-degree graphs); greedy matching on a sparse $G$ is $O(|E|\log n)$,
and local refinement is linear in tested pairs. Even for variable $n_i$, the search cost remains negligible relative to the constant-size SVDs dominating the Tucker step. The heuristic nature of the graph does not affect the monotonicity guarantees of Sec.~\ref{sec:convergence}; it primarily serves to \emph{accelerate} convergence by steering updates toward low inter-block correlation.

\paragraph{Remarks.}
(i) If prior structural knowledge (e.g., spatial locality, symmetries, or data layout) is available, it can be injected by masking or biasing edge weights.  
(ii) When hardware supports larger native gates, the same procedure extends to $n_i>2$ using $k$-set partitioning instead of matching; in current devices $k=2$ remains the most cost-effective.  
(iii) The graph may be recomputed at every iteration or updated incrementally; both choices preserve the monotonicity properties proven in Sec.~\ref{sec:convergence}.

\section{Numerical results}
\label{sec:results}

To evaluate the utility of the method, we apply it to the MNIST 784 dataset~\cite{lecun2010mnist}, compressing the circuit with a maximum factor size of two. Each digit is embedded into 10 qubits, and each iteration introduces five two-qubit general unitaries. These decompose into a circuit of approximate depth 13–14 with respect to the gate set $\{R_x, R_y, R_z, \mathrm{CX}\}$. Consequently, the circuit depth scales linearly with the number of iterations. For example, after six iterations the depth is approximately 78, while after 260 iterations it reaches about 3500, without applying any circuit optimization.

For the digit ``zero’’ (index 0 of~\cite{lecun2010mnist}), a clear fidelity-depth trade-off is observed: increasing the number of iterations reduces the fidelity loss, but simultaneously increases the circuit depth. This relationship is illustrated in Fig.~\ref{fig:progression-zero}. The standard Qiskit~\cite{Qiskit} initializer, which employs the isometries method~\cite{Iten2016}, is used as a baseline. Any circuit depth exceeding this baseline is marked as \textit{No Use}, as it provides no practical advantage. Approximations below this threshold -- even those associated with relatively high fidelity loss -- may still offer sufficient quality for certain applications.
\begin{figure}[ht]
    \centering
    \includegraphics[width=0.9\linewidth]{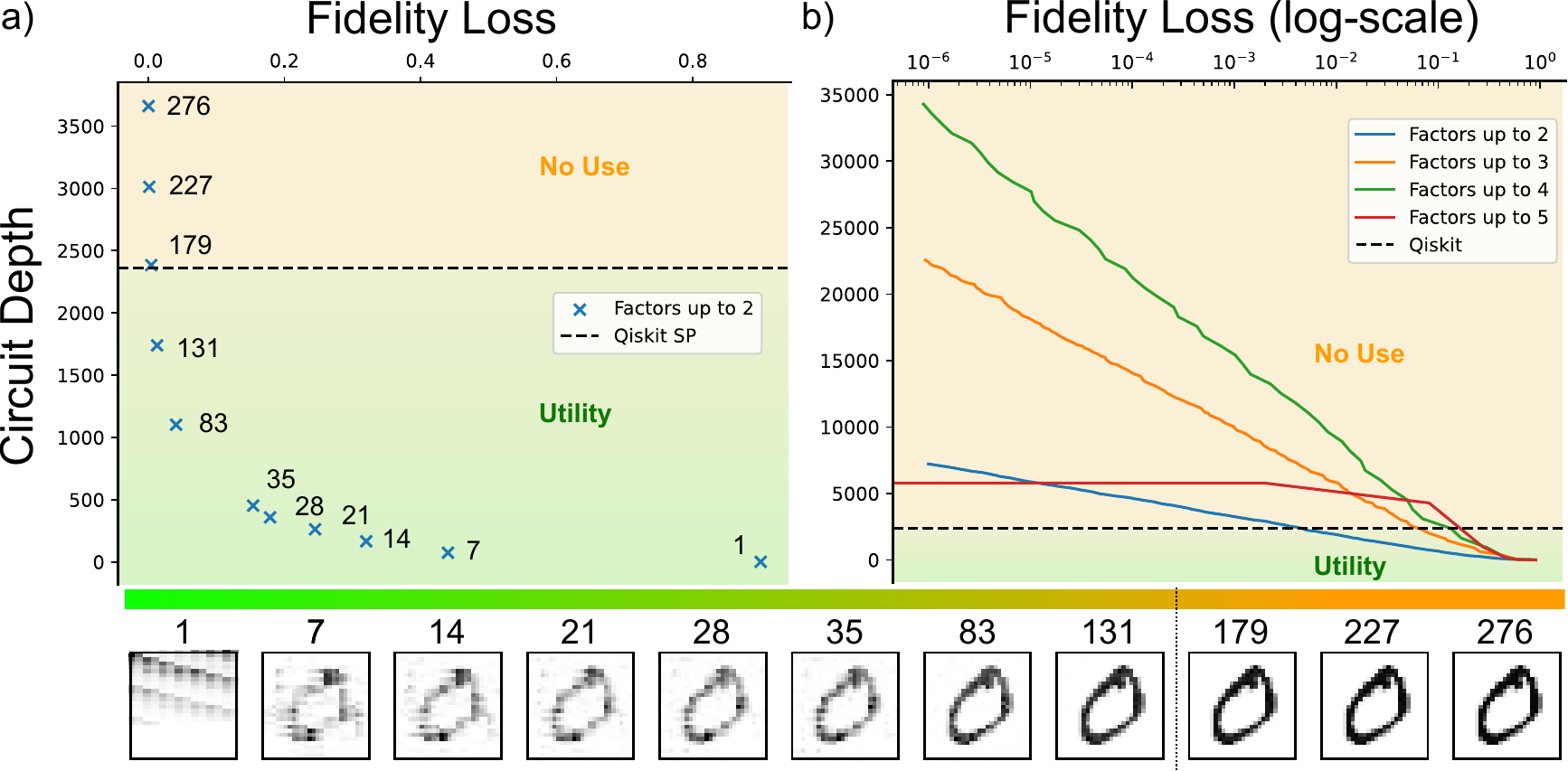}
    \caption{This study examines the relationship between iteration count, circuit depth, and fidelity loss in the Q-Tucker method. As illustrated by the scatter plot, the method achieves rapid convergence -- requiring fewer iterations -- when an acceptable fidelity loss is targeted, highlighting its potential utility.}
    \label{fig:progression-zero}
\end{figure}

To examine the effect of the factor size, values of 2, 3, 4, and 5 were compared on the same digit (cf. Fig.~\ref{fig:progression-zero}). Since the synthesis of general three-, four-, and five-qubit unitaries is expensive, the method does not scale efficiently in this regime, highlighting the need for improved multi-qubit unitary synthesis techniques. Nevertheless, larger factor sizes exhibit a clear advantage in terms of convergence speed (not shown in the figure). For the digit ``zero’’, convergence below a precision of $10^{-6}$ is achieved after 553 iterations with factor size 2, 163 with factor size 3, 51 with factor size 4, and only 8 with factor size 5. These results indicate that advances in unitary synthesis could yield substantial improvements for the iterative Tucker decomposition.

\section{Comments and Conclusion}
\label{sec:conclusion}

Empirical tests on real-world data show that even when Q-Tucker fails to fully converge -- an infrequent scenario -- it still produces high-quality approximations, often matching the best achievable results within a given classical computational budget. In our implementation, the number of iterations is capped at the square of the number of qubits. This parameter, which directly influences circuit depth, can be tuned to match the constraints of the target device and application. Within this limit, the method consistently delivers efficient decompositions compatible with available classical resources.

A key strength of the approach lies in its reliance on SVDs over two-qubit subsystems. Since these subsystems have fixed dimensionality, each SVD incurs a constant cost. This avoids the scalability bottleneck typical of tensor decompositions like Tucker, where SVD costs scale cubically with input size. Additionally, the heuristic used to search for a good configuration tuple has linear complexity with respect to the number of state amplitudes. As a result, the overall asymptotic cost of the method is linear in the number of amplitudes -- a significant improvement over prior approaches~\cite{araujo2024lowrank}.

The method also includes a convergence detection mechanism. When convergence is deemed unlikely, it falls back on the Schmidt method~\cite{araujo2024lowrank}, incorporating partial disentanglement. This concept of partial disentanglement enables compressions that standard methods often overlook. However, the fallback remains comparatively slow, even when combined with greedy heuristics, highlighting the need for faster alternatives.

In summary, Q-Tucker achieves the best possible approximation allowed by the classical resources at hand, offering two key guarantees: (1) The approximation improves proportionally with increased classical resources. (2) The algorithm always returns a result -- even for large, complex states where only a coarse approximation may be feasible.

\begin{credits}
\subsubsection{\ackname} We acknowledge funding of the Federal Ministry of Research, Technology and Space Project Number 13N17157 (Q-ROM).

\subsubsection{\discintname}
Carsten Blank is the co-founder and CEO of data cybernetics. Furthermore, a patent application covering the methods described herein has been filed and is pending; details are under confidentiality restrictions.

\end{credits}
%
%
%
\bibliographystyle{splncs04}
\bibliography{references}
\end{document}